\documentclass[submission,copyright,creativecommons]{eptcs}
\providecommand{\event}{TARK 2019} 
\usepackage{breakurl}             
\usepackage{underscore}           
\usepackage{graphicx}
\usepackage{amsthm}
\usepackage{amsfonts}
\usepackage{amssymb}
\usepackage{mathtools}
\usepackage{eucal}
\newcommand{\G}{Log_A \mathbf{G}}
\newcommand{\B}{Log_A \mathbf{B}}
\newcommand{\LS}{Log_A \mathbf{S}}
\newcommand{\g}{Log_A \mathbf{G}}
\newtheorem{definition}{Definition}[section]
\newtheorem{example}{Example}[section]

\newtheorem{theorem}{Theorem}[section]
\newtheorem{lemma}{Lemma}[section]
\newtheorem{obs}{Observation}[section]
\newtheorem{proposition}{Proposition}[section]
\newtheorem{cor}{Corollary}[section]

\title{A Unified Algebraic Framework for Non-Monotonicity}
\author{Nourhan Ehab
\institute{German University in Cairo\\ Department of Computer Science and Engineering\\ Cairo, Egypt}
\email{nourhan.ehab@guc.edu.eg}
\and
Haythem O. Ismail
\institute{Cairo University\\ Department of Engineering Mathematics\\
Cairo, Egypt}
\institute{German University in Cairo\\ Department of Computer Science and Engineering\\ Cairo, Egypt}
\email{haythem.ismail@guc.edu.eg}
}
\def\titlerunning{A Unified Algebraic Framework for Non-Monotonicity}
\def\authorrunning{Nourhan Ehab \& Haythem O. Ismail}

\newcommand{\interp}[1]{[\![ #1 ]\!]}
\newcommand{\gcon}{\mathrel{|\hspace{-0.45em}\simeq}}

\newcommand{\ivDash}{\mathrel{\reflectbox{\rotatebox[origin=c]{90}{$\vDash$}}}}

\usepackage{enumitem}
\usepackage{xcolor}
\setlist[description]{leftmargin=\parindent,labelindent=\parindent}

\begin{document}
\maketitle

\begin{abstract}
Tremendous research effort has been dedicated over the years to thoroughly investigate non-monotonic reasoning. With the abundance of non-monotonic logical formalisms, a unified theory that enables comparing the different approaches is much called for. In this paper, we present an algebraic graded logic we refer to as $\G$ capable of encompassing a wide variety of non-monotonic formalisms. We build on Lin and Shoham's argument systems first developed to formalize non-monotonic commonsense reasoning. We show how to encode argument systems as $\G$ theories, and prove that $\G$ captures the notion of belief spaces in argument systems. Since argument systems capture default logic, autoepistemic logic, the principle of negation as failure, and circumscription, our results show that $\G$ captures the before-mentioned non-monotonic logical formalisms as well. Previous results show that $\G$ subsumes possibilistic logic and any non-monotonic inference relation satisfying Makinson's rationality postulates. In this way, $\G$ provides a powerful unified framework for non-monotonicity.
\end{abstract}

\section{Introduction}
Non-monotonic logics are attempts to model commonsense defeasible reasoning that allows making plausible, albeit possibly fallible, assumptions about the world in the absence of complete knowledge. The term ``non-monotonic" refers to the fact that new evidence can retract previous contradicting assumptions. This contrasts with classical logics where new
evidence never invalidates previous assumptions about the world. Modelling non-monotonicity has been the focus of extensive studies in the knowledge representation and reasoning community for many years giving rise to a vast family of non-monotonic formalisms. The currently existing approaches to representing non-monotonicity can be classified into two orthogonal families: fixed point logics and model preference logics \cite{van2008handbook}. Fixed point logics define a fixed point operator by which possibly multiple sets of consistent beliefs can be constructed. Typical non-monotonic logics taking the fixed point approach are Reiter's default logic \cite{reiter1980logic} and Moore's autoepistemic logic \cite{moore1984possible,marek1991autoepistemic}. Model preference logics, on the other hand, define  non-monotonic logical inference relations with respect to selected preferred models of the world. Typical model preference logics are probabilistic logic \cite{adams1965logic,pearl2014probabilistic}, McCarthy's circumscription \cite{circumscription}, system P proposed by Kraus, Lehmann and Magidor \cite{kraus1990nonmonotonic}, and Pearl's system Z \cite{pearl1990system}. The wide diversity of all of these logics in addition to their non-standard semantics has rendered the task of gaining a good understanding of them quite hard. For this reason, a unified theory that enables comparing the different approaches is much called for. The purpose of this paper is to present an algebraic graded logic we refer to as $\G$ \cite{logag,nourhanThesis2016,logagcommonsense} capable of encompassing the previously-mentioned non-monotonic logics providing a general framework for non-monotonicity. 

Another non-standard attempt at formalizing commonsense non-monotonic reasoning is Lin and Shoham's argument systems \cite{lin1989argument}. Argument systems are considered a radical departure from the classical sentence-based approaches as they are based entirely on inference rules. In \cite{lin1989argument}, Lin and Shoham prove that classical non-monotonic approaches such as default logic \cite{reiter1980logic}, autoepistemic logic \cite{moore1984possible}, circumscription \cite{circumscription}, and the principle of negation as failure \cite{clark1978negation} are all special cases of argument systems. In this paper, we show that argument systems can be embedded in $\G$ proving that $\G$ captures the same non-monotonic logical approaches that argument systems capture. In \cite{logag-possibilistic}, we proved that $\G$ subsumes possibilistic logic \cite{possibilistic} and any non-monotonic inference relation satisfying Makinson's rationality postulates. While other unifying frameworks for non-monotonic formalisms such as \cite{bondarenko1997abstract,chen1993minimal} exist in the literature, non of these frameworks can capture weighted approaches such as possibilistic logic while encompassing the classical previously-mentioned logical approaches like $\G$ does. In this way, $\G$ can be considered a powerful algebraic unified framework for non-monotonicity providing a unified understanding of a vast diversity of non-monotonic logical formalisms.

The rest of this paper is organized as follows. In Section \ref{logag}, $\G$ will be thoroughly reviewed describing its syntax and semantics. Section \ref{argument-systems} will briefly review argument systems. In Section \ref{argument-systems-logag}, the main results of this paper, proving that $\G$ subsumes argument systems, will be presented. Finally, some concluding remarks are outlined in Section \ref{conc}.


\section{$\G$} \label{logag}
$\g$ is a graded logic for reasoning with uncertain beliefs. ``Log" stands for logic, ``A" for algebraic, and ``G" for grades. In $\g$, a classical logical formula could be associated with a grade representing a measure of its uncertainty. Non-graded formulas are taken to be certain. In this way, $\g$ is a logic for reasoning about graded propositions. $\g$ is algebraic in that it is a language of only terms, some of which denote propositions. Both propositions and their grades are taken as individuals in the $\g$ ontology. While some multimodal logics such as \cite{gradedtrust,condprob} may be used to express graded grading propositions too, unlike $\g$, the grades themselves are embedded in the modal operators and are not amenable to reasoning and quantification. This makes $\g$ a quite expressive language that is still intuitive and very similar in syntax to first-order logic. $\g$ is demonstrably useful in commonsense reasoning including default reasoning, reasoning with information provided by a chain of sources of varying degrees of trust, and reasoning with paradoxical sentences as discussed in \cite{nourhanThesis2016,logag}.

While most of the graded logics we are aware of employ non-classical modal logic semantics by assigning grades to possible worlds \cite{dubois2014weighted}, $\g$ is a non-modal logic with classical notions of worlds and truth values. This is not to say that $\g$ is a classical logic, but it is closer in spirit to classical non-monotonic logics such as default logic and circumscription. Following these formalisms, $\g$ assumes a classical notion of logical consequence on top of which a more restrictive, non-classical relation is defined selecting only a subset of the classical models. In defining this relation we take the algebraic, rather than the modal, route. The remaining of this section is dedicated to reviewing the syntax and semantics of $\G$. A sound and complete proof theory for $\G$ is presented in \cite{logag,nourhanThesis2016}. In \cite{logag}, it was proven that $\G$ is a stable and well-behaved logic observing Makinson's properties of reflexivity, cut, and cautious monotony.

\subsection{$\G$ Syntax}
$\G$ consists of algebraically constructed terms from function symbols. There are no sentences in $\G$; instead, we use terms of a distinguished syntactic type to denote propositions. $\G$ is a variant of $\B$ \cite{logab} and $\LS$ \cite{logas}, which are algebraic
languages for reasoning about, respectively, beliefs and temporal phenomena.
Propositions are included as first-class individuals in the $\G$ ontology and are structured in a Boolean algebra giving us all standard truth conditions and classical notions of consequence and validity. 
The inclusion of propositions in the ontology, though non-standard, has been suggested by several authors \cite{church1950carnap,bealer1979theories,parsons1993denoting,shapiro1993belief}. We refer the reader to \cite{logab,shapiro1993belief} for arguments in favour of adopting this approach in the representation of propositional attitudes in artificial intelligence.  Additionally, \emph{grades} are introduced as first-class individuals in the ontology. As a result, propositions \emph{about} graded propositions can be constructed, which are themselves recursively gradable. 
                
A $\G$ language is a many-sorted language composed of a set of terms partitioned into three base sorts: $\sigma_P$ is a set of terms (including the term $\mathbf{true}$)  denoting propositions, $\sigma_D$ is a set of terms denoting grades, and $\sigma_I$ is a set of terms denoting anything else.
A \emph{$\G$ alphabet} includes a non-empty, countable set of constant and function symbols each having a syntactic sort from the set $\sigma= \{\sigma_P,
\sigma_D, \sigma_I\} \cup \{\tau_1 \longrightarrow \tau_2 ~|~ \tau_1
\in \{\sigma_P,\sigma_D, \sigma_I\}\}$ and $\tau_2 \in \sigma\}$ of syntactic sorts. Intuitively, $\tau_1 \longrightarrow \tau_2$ is
the syntactic sort of function symbols that take a single argument
of sort $\sigma_P$, $\sigma_D$, or $\sigma_I$ and produce a functional term of sort $\tau_2$. Given the restriction of the first argument of function symbols to base sorts, $\G$ is, in a sense, a first-order language. In addition, an alphabet includes a countably infinite set of variables of the three base sorts; a set of syncategorematic symbols including the comma, various matching pairs of brackets and parentheses, and the symbol $\forall$; and a set of logical symbols defined as the union of the following sets: $\{\neg\} \subseteq \sigma_P \longrightarrow \sigma_P$, $\{\wedge, \vee\} \subseteq \sigma_P \longrightarrow \sigma_P \longrightarrow\sigma_P$, $\{\prec, \doteq\} \subseteq \sigma_D \longrightarrow \sigma_D \longrightarrow \sigma_P$, and $\{\mathbf{G}\} \subseteq \sigma_P \longrightarrow \sigma_D \longrightarrow \sigma_P$. 
Terms involving $\supset$ \footnote{Through out this paper, we will use $\supset$ to denote material implication.} and $\exists$ can always be expressed in terms of the above logical operators and $\forall$. The terms containing $\mathbf{G}$ are referred to \emph{grading terms}, while the terms not including $\mathbf{G}$ are referred to as \emph{non-grading terms}.  

The following are some examples of well-formed $\sigma_P$ terms permissible by the syntax of $\g$.
\begin{enumerate}
\item $\forall d_1, d_2 [d_1 \doteq d_2 \supset d_2 \doteq d_1]$
    \item $\forall d_1, d_2 [\neg (d_1 \prec d_2) \Leftrightarrow (d_2 \prec d_1 \vee d_2 \doteq d_1)]$
  \item $\mathbf{G}(P,2)$
    \item $\forall x [P(x) \supset \mathbf{G}(Q(x),5)]$
    \item $\mathbf{G}(\forall x [P(x) \supset \neg Q(x)],10)$
     \item $\mathbf{G}(\mathbf{G}(\mathbf{G}(R, 2), 3),12)$
\end{enumerate}

The first two well-formed terms denote properties of grades: (1) denotes the proposition that the equality relation of grades is symmetric, and (2) denotes the proposition that the ordering of grades is linear.  (3) denotes the proposition that the grade of $P$ is 2. (4) and (5) illustrate the \emph{de re} and \emph{de dicto} grading, respectively. The syntax of $\g$ allows the nesting of grading terms forming grading chains as shown in (6). One possible use of such nesting is to express information from various knowledge sources with different trust degrees \cite{logag}.   
\newpage
\subsection{From Syntax to Semantics}
A key element in the semantics of $\G$ is the notion of a \emph{$\G$} structure.

\begin{definition}
A \emph{$\G$ structure} is a quintuple $\mathfrak{S}=\langle
\mathcal{D}, \mathfrak{A}, \mathfrak{g}, <, \mathfrak{e}\rangle$,
where
\begin{itemize}
    \itemsep0pt
    \item $\mathcal{D}$, the domain of discourse, is a set with two disjoint,
    non-empty, countable subsets: a set of propositions $\mathcal{P}$, and a set of grades $\mathcal{G}$. 
    \item $\mathfrak{A} = \langle \mathcal{P}, +, \cdot, -, \bot, \top\rangle$ is a complete, non-degenerate Boolean algebra.
    \item $\mathfrak{g} : \mathcal{P} \times \mathcal{G} \longrightarrow
    \mathcal{P}$ is a grading function.
    \item $< : \mathcal{G} \times \mathcal{G} \longrightarrow
    \mathcal{P}$ is an ordering function.
    \item $\mathfrak{e} : \mathcal{G} \times \mathcal{G} \longrightarrow
    \{\bot,\top\}$ is an equality function, where for every $g_1, g_2 \in \mathcal{G}$:\\ $\mathfrak{e}(g_1,g_2) = \top$ if $g_1=g_2$, and $\mathfrak{e}(g_1,g_2) = \bot$ otherwise.
\end{itemize}
\end{definition}

\vspace*{0px}

A \emph{valuation} $\mathcal{V}$ of a $\G$ language is
a triple $\langle \mathfrak{S}, \mathcal{V}_f,
\mathcal{V}_x\rangle$, where $\mathfrak{S}$ is a $\G$ structure, $\mathcal{V}_f$ is a function that assigns to each function symbol an appropriate function on $\mathcal{D}$, and $\mathcal{V}_x$ is a function mapping each variable to a corresponding element of the appropriate block of $\mathcal{D}$.
An \emph{interpretation} of $\G$ terms is given by a function $\interp{\cdot}^{\mathcal{V}}$.  
Figure \ref{fig:interp} summarizes the operation of $\interp{\cdot}^{\mathcal{V}}$.
\begin{figure}[h]
\centering
    \includegraphics[width=0.45\linewidth]{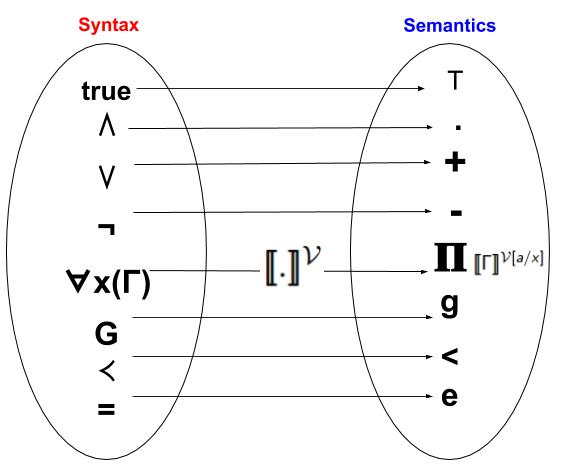}
    \caption{\footnotesize{The interpretation of the $\g$ terms.}}
    \label{fig:interp}
\end{figure}

\begin{definition}
Let $L$ be a $\G$ language and let $\mathcal{V}$ be a
valuation of $L$. An \emph{interpretation} of the terms of
$L$ is given by a function $\interp{\cdot}^{\mathcal{V}}$:
\begin{itemize}
    \itemsep0pt
     \item $\interp{true}^{\mathcal{V}} = \top$
    \item $\interp{x}^{\mathcal{V}} = \mathcal{V}_x(x)$, for a variable $x$
    \item $\interp{c}^{\mathcal{V}} = \mathcal{V}_f(c)$, for a constant $c$
    \item $\interp{f(t_1, \ldots, t_n)}^{\mathcal{V}} = \mathcal{V}_f(f)(\interp{t_1}^{\mathcal{V}}, \ldots, \interp{t_n}^{\mathcal{V}})$, for an $n$-adic ($n \geq 1$) function symbol $f$
    \item $\interp{(t_1 \wedge t_2)}^{\mathcal{V}} = \interp{t_1}^{\mathcal{V}} \cdot \interp{t_2}^{\mathcal{V}}$
    \item $\interp{(t_1 \vee t_2)}^{\mathcal{V}} = \interp{t_1}^{\mathcal{V}} + \interp{t_2}^{\mathcal{V}}$
    \item $\interp{\neg t}^{\mathcal{V}} = -\interp{t}^{\mathcal{V}}$
    \item $\interp{\forall x(t)}^{\mathcal{V}} = \displaystyle \prod_{a \in \mathcal{D}} \interp{t}^{\mathcal{V}[a/x]}$
    \item $\interp{\mathbf{G}(t_1,t_2)}^{\mathcal{V}} = \mathfrak{g}(\interp{t_1}^{\mathcal{V}},\interp{t_2}^{\mathcal{V}})$
    \item $\interp{t_1 \prec t_2}^{\mathcal{V}} = \interp{t_1}^{\mathcal{V}} < \interp{t_2}^{\mathcal{V}}$
    \item $\interp{t_1 \doteq t_2}^{\mathcal{V}} = \mathfrak{e}(\interp{t_1}^{\mathcal{V}},\interp{t_2}^{\mathcal{V}})$
\end{itemize}
\end{definition}
 
\subsection{Beyond Classical Logical Consequence}
We define logical consequence using the familiar notion of filters from Boolean algebra \cite{algebra}.
\begin{definition}\label{filter}
A filter of a boolean algebra $\mathfrak{A}=\langle \mathcal{P}, +, \cdot, -, \bot, \top\rangle$ is a subset $F$ of $\mathcal{P}$ such that:
\begin{enumerate}
    \item  $\top \in F$;
    \item If $a, b \in F $, then $a \cdot b \in F$; and
    \item  If $a \in F$ and $a \le b$, then $b \in F$.
\end{enumerate}
\end{definition}

A propositional term $\phi$ is a logical consequence of a set of propositional terms $\Gamma$ if it is a member of the filter of the interpretation of $\Gamma$, denoted $F(\interp{\Gamma}^\mathcal{V})$. 
\begin{definition}\label{def:con2}
Let $L$ be a $\G$ language.  For every $\phi \in \sigma_P$
and ${\rm \Gamma} \subseteq \sigma_P$, $\phi$ is a logical
consequence of ${\rm \Gamma}$, denoted ${\rm \Gamma} \models
\phi$, if, for every $L$-valuation $\mathcal{V}$,
$\interp{\phi}^\mathcal{V} \in F(\interp{\Gamma}^\mathcal{V})$ where $\interp{\Gamma}^\mathcal{V}=\displaystyle\prod_{\gamma \in {\rm \Gamma}}
\interp{\gamma}^{\mathcal{V}}.$
\end{definition}

Unfortunately, the definition of logical consequence presented in the previous definition cannot address uncertain reasoning with graded propositions. To see that, consider the following situation. You see a bird from far away that looks a lot like a penguin. You know that any penguin has wings but does not fly. To make sure that what you see is indeed a penguin, you ask your brother who tells you that  
this bird must not be a penguin since your sister told him that she saw the same bird flying.
This situation can be represented in $\G$  by a set of propositions $\mathcal{Q}$ as shown in Figure \ref{fig:gradedFilters} where $p$ denotes that the bird is a penguin, $w$ denotes has wings, and $f$ denotes that the bird flies. For the ease of readability of Figure \ref{fig:gradedFilters}, we write $\neg \phi$ instead of $-\phi$ and $\phi \supset \psi$ instead of $-\phi+\psi$. Since you are uncertain about whether the bird you see is a penguin, this is represented as a graded proposition $\mathfrak{g}(p,d1)$ where $d1$ is your uncertainty degree in what you see. What your brother tells you is represented by the grading chain $\mathfrak{g}(\mathfrak{g}(f,d2),d3)$ where $d3$ represents how much you trust your brother, and $d2$ represents how much you trust your sister. Now, consider an agent reasoning with the set $\mathcal{Q}$. 
Initially, it would make sense for the agent to be able to conclude $p$ even if $p$ is uncertain (and, hence, graded) since it has no reason to believe $\neg p$. The filter $F(\mathcal{Q})$, however, contains the classical logical consequences of $\mathcal{Q}$, but will never contain the graded proposition $p$.
For this reason, we extend our classical notion of filters into a more liberal notion of \emph{graded filters} to enable the agent to conclude, in addition to the classical logical consequences of $\mathcal{Q}$,  propositions that are graded in $\mathcal{Q}$ (like $p$) or follow from graded propositions in $\mathcal{Q}$ (like $\neg f$ and $w$). This should be done without introducing inconsistencies. Due to nested grading, graded filters come in degrees depending on the depth of nesting of the admitted graded propositions. In Figure \ref{fig:gradedFilters}, $\mathcal{F}^1(\mathcal{Q})$ is the graded filter of degree 1. $\mathcal{F}^1(\mathcal{Q})$ contains everything in $F(\mathcal{Q})$ in addition to the nested graded propositions at depth 1, $p$ and $\mathfrak{g}(f, d2)$. $\neg f$ and $w$ are also admitted to $\mathcal{F}^1(\mathcal{Q})$ since they follow classically from $\{p,p\supset \neg f\}$ and $\{p,p\supset w\}$ respectively. Consequently, at degree 1, we end up believing that the bird is a penguin that has wings and does not fly. To compute the graded filter of degree 2, $\mathcal{F}^2(\mathcal{Q})$, we take everything in $\mathcal{F}^1(\mathcal{Q})$ and try to add the graded proposition $f$ at depth 2. The problem is, once we do that, we have a contradiction with $\neg f$ (we now believe that bird flies and does not fly at the same time). To resolve the contradiction, we admit to $\mathcal{F}^2(\mathcal{Q})$ either $p$ (and consequently $\neg f$ and $w$) or $f$. In deciding which of $p$ or $f$ to kick out we will allude to their grades. The grade of $p$ is $d1$, and $f$ is graded in a grading chain containing $d2$ and $d3$. To get a fused grade for $f$, we will combine both $d2$ and $d3$ using an appropriate fusion operator. If $d1$ is less than the fused grade of $f$, $p$ will not be admitted to the graded filter, together with it consequence $\neg f$. Otherwise, $f$ will not be admitted, and $p$ and $\neg f$ will remain. If we try to compute $\mathcal{F}^3(\mathcal{Q})$, we get everything in $\mathcal{F}^2(\mathcal{Q})$ reaching a fixed point. 

\begin{figure} [h]
\centering
    \includegraphics[width=0.5\linewidth]{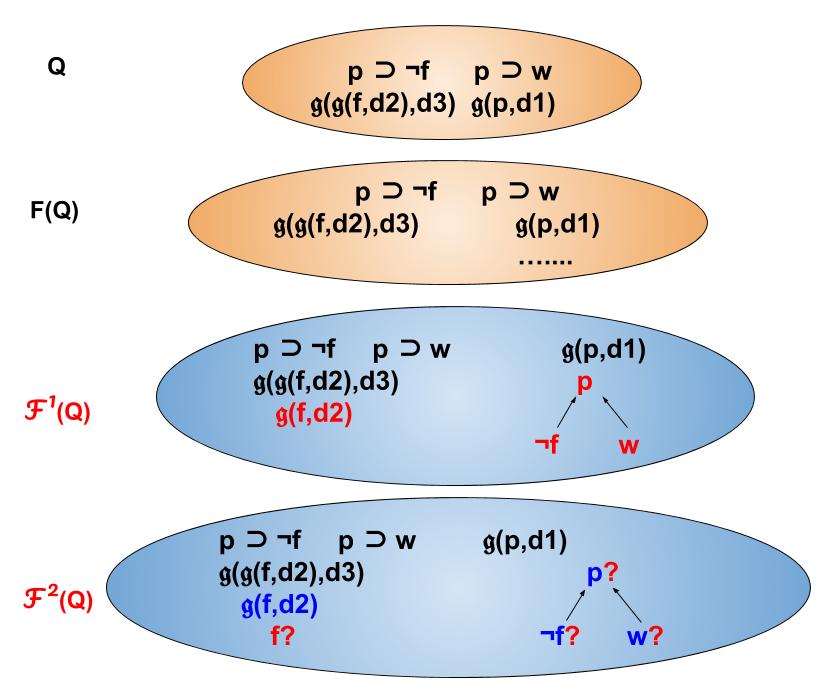}
    \caption{\footnotesize{Graded Filters}}
    \label{fig:gradedFilters}
\end{figure}

In general, the elements of $\mathcal{F}^i(\mathcal{Q})$ will be referred to as the graded consequences at level $i$. The rest of this section is dedicated to formally defining graded filters together with our graded consequence relation based on graded filters. In the sequel, for every $p \in \mathcal{P}$ and $g \in
\mathcal{G}$, $\mathfrak{g}(p, g)$ will be taken to represent a \emph{grading proposition} that \emph{grades} $p$.
Moreover, if $\mathfrak{g}(p, g) \in \mathcal{Q} \subseteq
\mathcal{P}$, then $p$ is \emph{graded} in $\mathcal{Q}$.
The set of \emph{$p$ graders in $\mathcal{Q}$} is defined to be the
set $Graders(p, \mathcal{Q}) = \{q | q \in \mathcal{Q}$ and $q$ grades
$p\}$. Throughout, a $\G$ structure
$\mathfrak{S}=\langle \mathcal{D}, \mathfrak{A}, \mathfrak{g}, <,
\mathfrak{e}\rangle$ is assumed.

As a building step towards formalizing the notion of a graded filter, the structure of graded propositions should be carefully specified. For this reason, the following notion of an \emph{embedded proposition} is defined. 
\begin{definition}\label{def:embed}
Let $\mathcal{Q} \subseteq \mathcal{P}$. A proposition $p \in
\mathcal{P}$ is embedded in $\mathcal{Q}$ if (i) $p \in
\mathcal{Q}$ (ii) or if, for some $g \in \mathcal{G}$,
$\mathfrak{g}(p, g)$ is embedded in $\mathcal{Q}$. Henceforth, let
$E(\mathcal{Q}) = \{p | p $ is embedded in $\mathcal{Q}\}$.
\end{definition}

Since a graded proposition $p$ might be embedded at any depth $n \in \mathbb{N}$, the degree of embedding of a graded proposition $p$ is defined as follows.

\begin{definition}\label{def:degembed}
For $\mathcal{Q} \subseteq \mathcal{P}$, let the degree of embedding of $p$ in $\mathcal{Q}$ be a function  $\delta_{\mathcal{Q}}
: E(\mathcal{Q}) \longrightarrow \mathbb{N}$, where
\begin{enumerate}\itemsep0pt
    \itemsep0pt
    \item if $p \in \mathcal{Q}$, then $\delta_{\mathcal{Q}}(p) = 0$; and
    \item if $p \notin \mathcal{Q}$, then $\delta_{\mathcal{Q}}(p) =
    e+1$, where $e = \min_{q \in Graders(p, E(\mathcal{Q}))} \{\delta_{\mathcal{Q}}(q)\}$.
\end{enumerate}
\end{definition}
For notational convenience, we let the set of embedded propositions at depth $n$ be $E^n(\mathcal{Q}) = \{p \in
E(\mathcal{Q})~|~\delta_\mathcal{Q}(p) \leq n\}$, for every $n \in
\mathbb{N}$.


The key to defining graded filters is the intuition that the set of consequences of a proposition set $\mathcal{Q}$ may be further enriched by \emph{telescoping} $\mathcal{Q}$ and accepting some of the propositions graded therein. We refer to this process as telescoping as the set of graded filters at increasing depths looks like an inverted telescope (as illustrated in Figure \ref{fig:gradedFilters}). For this, we need to define (i) the process of telescoping, which is a step-wise process that considers propositions at increasing grading depths, and (ii) a criterion for accepting graded propositions which, as mentioned before, depends on the grades of said propositions. Since the nesting of grading chains is permissible in $\G$, it is necessary to compute the \emph{fused grade} of a graded proposition $p$ in a chain $C$ to decide whether it will be accepted or not. The fusion of grades in a chain is done according to an operator $\otimes$. Further, since a graded proposition $p$ might be graded by more than one grading chain, we define the notion of the fused grade of $p$ across all the chains grading it by an operator $\oplus$.

\begin{definition}\label{def:telstruct}
Let $\mathfrak{S}$ be a $\G$ structure with a depth- and
fan-out-bounded $\mathcal{P}$ \footnote{$\mathcal{P}$ is depth-bounded if every grading chain has at most $d$ distinct grading propositions and is fan-out-bounded if every grading proposition grades at most $f_{out}$ propositions where $d, f_{out} \in \mathbb{N}$ \cite{nourhanThesis2016}.}. A telescoping structure for
$\mathfrak{S}$ is a quadruple $\mathfrak{T}=\langle \mathcal{T},
\mathfrak{O}, \otimes, \oplus\rangle$, where
\begin{itemize}
\itemsep0pt
\item $\mathcal{T} \subseteq \mathcal{P}$, referred to as the top theory; %
\item $\mathfrak{O}$ is an ultrafilter of the subalgebra induced
by $Range(<)$ (an ultrafilter is a maximal filter with respect to not including $\bot$) \cite{algebra}; 
\item $\otimes : \bigcup_{i=1}^\infty \mathcal{G}^i
\longrightarrow \mathcal{G}$; and $\oplus: \bigcup_{i=1}^\infty \mathcal{G}^i \longrightarrow
\mathcal{G}$. 
\end{itemize}
\end{definition}

Recasting the familiar notion of a \emph{kernel} of a belief base \cite{hansson} into the context of $\G$ structures, we say that a
$\bot$-kernel of $\mathcal{Q} \subseteq \mathcal{P}$ is a subset-minimal inconsistent set
$\mathcal{X} \subseteq \mathcal{Q}$ such that $F(E(F(\mathcal{X})))$ is
improper ($=\mathcal{P}$) where $E(F(\mathcal{X}))$ is the set of all embedded graded propositions in the filter of $\mathcal{X}$. Let $\mathcal{Q}^{\ivDash} \bot$ be the set of $\mathcal{Q}$ kernels that entail $\bot$. A proposition $p \in \mathcal{X}$ \emph{survives} $\mathcal{X}$ in $\mathfrak{T}$ if $p$ is not the weakest proposition (with the least grade) in $\mathcal{X}$. In what follows, the fused grade of a proposition $p$ in $\mathcal{Q}\subseteq \mathcal{P}$ according to a telescoping structure $\mathfrak{T}$ will be referred to as $\mathfrak{f}_\mathfrak{T}(p,\mathcal{Q})$.

\begin{definition}\label{def:surviving}
For a telescoping structure $\mathfrak{T}=\langle \mathcal{T},
\mathfrak{O}, \otimes, \oplus\rangle$ and a fan-in-bounded \footnote{$\mathcal{Q}$ is fan-in-bounded if every graded proposition is graded by at most $f_{in}$ grading propositions where $f_{in} \in \mathbb{N}$ \cite{nourhanThesis2016}.}
$\mathcal{Q}\subseteq \mathcal{P}$, if $\mathcal{X} \subseteq
\mathcal{Q}$, then $p \in \mathcal{X}$ \emph{survives}
$\mathcal{X}$ \emph{given} $\mathfrak{T}$ if
\begin{enumerate}\itemsep0pt
\itemsep0pt
    \item $p$ is ungraded in $\mathcal{Q}$; or
    \item there is some ungraded $q \in \mathcal{X}$ such that $q \notin F(\mathcal{T})$; or
    \item there is some graded $q \in \mathcal{X}$ such that $q \notin F(\mathcal{T})$ and $(\mathfrak{f}_\mathfrak{T}(q,\mathcal{Q}) < \mathfrak{f}_\mathfrak{T}(p,\mathcal{Q})) \in \mathfrak{O}$.
\end{enumerate}
The set of kernel survivors of $\mathcal{Q}$ given $\mathfrak{T}$ is
the set
\begin{center}
$\kappa(\mathcal{Q}, \mathfrak{T}) = \{p \in \mathcal{Q}~|~$ if  $p
\in \mathcal{X} \in \mathcal{Q}^{\ivDash} \bot$ then $p$ survives
$\mathcal{X}$ given $\mathfrak{T}\}$.
\end{center}
\end{definition}

The notion of a proposition $p$ being \emph{supported} in $\mathcal{Q}$ is defined as follows.

\begin{definition}\label{def:support}
Let $\mathcal{Q}, \mathcal{T} \subseteq \mathcal{P}$. We say that
$p$ is supported in $\mathcal{Q}$ given $\mathcal{T}$ if
\begin{enumerate}\itemsep0pt
\itemsep0pt
    \item $p \in F(\mathcal{T})$; or
    \item there is a grading chain $\langle q_0, q_1, \ldots,
    q_{n}\rangle$ of $p$ in $\mathcal{Q}$ with $q_0 \in F(\mathcal{R})$ where every member of $\mathcal{R}$ is supported
    in $\mathcal{Q}$.
\end{enumerate}
The set of propositions supported in $\mathcal{Q}$ given
$\mathcal{T}$ is denoted by $\varsigma(\mathcal{Q}, \mathcal{T})$.
\end{definition}

\begin{obs} \label{supported-filter}
$\varsigma(\mathcal{Q}, \mathcal{T}) = F(\mathcal{T}) \cup EG$, for some set $EG$ of embedded graded propositions in $\mathcal{Q}$.
\end{obs}

The $\mathfrak{T}$-induced telescoping of $\mathcal{Q}$ is defined as the set of propositions supported given $\mathcal{T}$ in the set of kernel survivors of $E^1(F(\mathcal{Q}))$.

\begin{definition}\label{def:telescope}
Let $\mathfrak{T}$ be a telescoping structure for $\mathfrak{S}$.
If $\mathcal{Q} \subset \mathcal{P}$ such that
$E^1(F(\mathcal{Q}))$ is fan-in-bounded, then the
$\mathfrak{T}$-induced telescoping of $\mathcal{Q}$ is given by
\begin{center}
$
\tau_\mathfrak{T}(\mathcal{Q}) =
\varsigma(\kappa(E^1(F(\mathcal{Q})), \mathfrak{T}), \mathcal{T}).
$
\end{center}
\end{definition}

\begin{obs}\label{obs:consistency}
If $F(\mathcal{T})$ is proper, then $F(\varsigma( \kappa(\mathcal{Q},
\mathfrak{T}), \mathcal{T}))$ is proper.
\end{obs}

\begin{proof}
If $F(\varsigma( \kappa(\mathcal{Q},
\mathfrak{T}), \mathcal{T}))$ is not proper, then
$\varsigma( \kappa(\mathcal{Q},
\mathfrak{T}), \mathcal{T})$ has at least one kernel
$\mathcal{X} \in \mathcal{Q}^{\ivDash} \bot$. According to Definitions \ref{def:surviving} and \ref{def:support}, this can only happen if $\mathcal{X} \subseteq \mathcal{T}$. Thus, $F(\mathcal{T})$ is
proper.  
\end{proof}
\begin{definition}\label{def:gfilter}
If $F(\mathcal{Q})$ has finitely-many grading propositions, then
$\tau_\mathfrak{T}(\mathcal{Q})$ is defined, for every telescoping structure $\mathfrak{T}$. 
Hence, provided that the right-hand side is defined, let
$$
\tau_\mathfrak{T}^n(\mathcal{Q}) = \left\{
\begin{array}{ll}
  \mathcal{Q} & \mathrm{if~} n = 0 \\
  \tau_\mathfrak{T}(\tau_\mathfrak{T}^{n-1}(\mathcal{Q})) & \mathrm{otherwise}\\
\end{array}
\right.
$$

A \emph{graded filter} of a top theory $\mathcal{T}$, denoted
$\mathcal{F}^n(\mathfrak{T})$, is defined as the filter of the  $\mathfrak{T}$-induced telescoping of $\mathcal{T}$ of degree $n$. 
\end{definition}

In the following example, we now go back to the example we introduced at the beginning of this section in Figure \ref{fig:gradedFilters}. We show how the formal construction of the graded filters matches the intuitions we pointed out earlier. 

\begin{example}\label{ex:logag}
\normalfont{
Consider $\mathcal{Q}=\{ -p+-f,~-p+w,~\mathfrak{g}(p,2), \mathfrak{g}(\mathfrak{g}(f,2), 3)\}$ and $\mathfrak{T}=\langle \mathcal{Q},  \mathfrak{O}, \otimes, \oplus\rangle$ where $\oplus=max$, and $\otimes=mean$. In what follows, let $\tau^n_\mathfrak{T}$ be an abbreviation for $\tau^n_\mathfrak{T}(\mathcal{Q})$. 
\begin{itemize}[leftmargin=*]
\itemsep0pt
\item $\tau^0_\mathfrak{T}~=~\mathcal{Q}$ 
\item  $\tau^1_\mathfrak{T}~$ = $\tau_\mathfrak{T}(\tau^0_\mathfrak{T}) = \varsigma(\kappa(E^1(F(\mathcal{Q})), \mathfrak{T}), \mathcal{T})$ 
\vspace*{2px}
\begin{center}
\begin{tabular}{ |p{3.5cm} | p{8cm}|} \hline
$F(\mathcal{Q})$  & $\mathcal{Q}\cup\{-p+-f . \mathfrak{g}(p,2),~\mathfrak{g}(\mathfrak{g}(f,2),3) + \mathfrak{g}(p,2),...\}$ \\ \hline 
$E^1(F(\mathcal{Q})$ & $F(\mathcal{Q})~\cup~\{p,~ \mathfrak{g}(f,2)\}$ \\ \hline
$\kappa(E^1(F(\mathcal{Q}),\mathfrak{T}))$ & $F(\mathcal{Q})~\cup~\{p,~ \mathfrak{g}(f,2)\}$ \\  \hline 
$\varsigma(\kappa(E^1(F(\mathcal{Q}),\mathfrak{T}), \mathcal{Q})$ & $F(\mathcal{Q})~\cup~\{p,~ \mathfrak{g}(f,2)\}$ \\ \hline
$\mathcal{F}^1(\mathfrak{T})$ & $ F(\tau^1_\mathfrak{T}) = \{p,w,-f,...\}$ \\ \hline
\end{tabular}
\end{center}

Upon telescoping to degree 1, there are no contradictions in $E^1(F(\mathcal{Q}))$ (no $\bot-$kernel $\mathcal{X} \subseteq E^1(F(\mathcal{Q}$)). Hence, everything in $E^1(F(\mathcal{Q}))$ survives telescoping and is supported (notice the equality of $E^1(F(\mathcal{Q}))$, the kernel survivors, and the supported propositions in $E^1(F(\mathcal{Q}))$). At level 1, we believe that the bird we saw is indeed a penguin and accordingly has wings and does not fly.
\item  $\tau^2_\mathfrak{T}~=~\tau_\mathfrak{T}(\tau^1_\mathfrak{T})=~\varsigma(\kappa(E^1(F(\tau^1_\mathfrak{T})),\mathfrak{T}),\mathcal{Q})$
\vspace*{2px}
\begin{center}
\begin{tabular}{ |p{3.5cm} | p{8cm}|}\hline
$E^1(F(\tau^1_\mathfrak{T}))$  & $F(\tau^1_\mathfrak{T}) \cup \{f\}$ \\ \hline 
$\kappa(E^1(F(\tau^1_\mathfrak{T})),\mathfrak{T})$ & $E^1(F(\tau^1_\mathfrak{T})) - \{ p\}$ \\ \hline
$\varsigma(\kappa(E^1(F(\tau^1_\mathfrak{T})),\mathfrak{T}),\mathcal{Q})$ & $\kappa(E^1(F(\tau^1_\mathfrak{T})),\mathfrak{T})-\{-f,w\}$ \\ \hline
$\mathcal{F}^2(\mathfrak{T})$ & $F(\tau^2_\mathfrak{T})$ \\ \hline
\end{tabular}
\end{center}

Upon telescoping to degree 2, there are two $\bot-$kernels $\{f, -f\}$ and $\{ p, -p+-f, f\}$. $-f$ survives the first kernel as it is not graded in $\mathcal{Q}$. $f$ survives the first kernel as well as it is the only graded proposition in the kernel with another member $-f \notin F(\mathcal{Q})$. $p$ does not survive the second kernel as the kernel contains another graded proposition $f$ and the grade of $p$ (2) is less than the fused grade of $f$ $(mean(2,3) = 2.5)$. Accordingly, $-f$ loses its support and is not supported in the set of kernel survivors. The graded filter of degree 2 does not contain $p$ or $-f$, but $w$ is retained as it has nothing to do with the contradiction. At level 2, we start taking into account the information our brother told us. Since our combined trust in our brother and sister is higher that our trust in what we saw, we end up not believing that the bird we saw is a penguin since we believe that it flies.
\item  $\tau^3_\mathfrak{T} = \tau_\mathfrak{T}(\tau^2_\mathfrak{T})$\\
$\tau^3_\mathfrak{T}=\varsigma(\kappa(E^1(F(\tau^2_\mathfrak{T})),\mathfrak{T}),\mathcal{Q}) = \kappa(E^1(F(\tau^2_\mathfrak{T})),\mathfrak{T})$\\
$\mathcal{F}^3(\mathfrak{T}) = F(\tau^3_\mathfrak{T}) = \mathcal{F}^2(\mathfrak{T})$ reaching a fixed point.\qed
\end{itemize}  
}
\end{example}

Henceforth, given a $\G$ theory $\mathbb{T} \subseteq \sigma_P$ and a valuation $\mathcal{V} =
\langle \mathfrak{S}, \mathcal{V}_f,
\mathcal{V}_x\rangle$, let the valuation of $\mathbb{T}$ be denoted as $\mathcal{V}(\mathbb{T}) =
\{\interp{p}^\mathcal{V}~|~p \in \mathbb{T}\}$. We use graded filters to define graded consequence as follows. Further,
for a $\G$ structure $\mathfrak{S}$, an $\mathfrak{S}$ grading
canon is a triple $\mathcal{C} = \langle \otimes, \oplus,n\rangle$
where $n \in \mathbb{N}$ and $\otimes$ and $\oplus$ are as
indicated in Definition \ref{def:telstruct}.
\begin{definition}\label{def:cons}
Let $\mathbb{T}$ be a $\G$ theory. 
For every $p \in \sigma_P$, valuation $\mathcal{V}= \langle
\mathfrak{S}, \mathcal{V}_f, \mathcal{V}_x\rangle$ where $\mathfrak{S}$ has a set $\mathcal{P}$ which is
depth- and fan-out-bounded, and $\mathfrak{S}$ grading canon
$\mathcal{C} = \langle \otimes, \oplus,n\rangle$, $p$ is a
graded consequence of $\mathbb{T}$ with respect to $\mathcal{C}$,
denoted $\mathbb{T} \gcon^\mathcal{C} p$, if
$\mathcal{F}^n(\mathfrak{T})$ is defined and
$\interp{p}^\mathcal{V} \in \mathcal{F}^n(\mathfrak{T})$ for
every telescoping structure $\mathfrak{T} = \langle
\mathcal{V}(\mathbb{T}), \mathfrak{O},\otimes, \oplus\rangle$ for
$\mathfrak{S}$ where $\mathfrak{O}$ extends
$F(\mathcal{V}(\mathbb{T}) \cap Range(<))$\footnote{An
ultrafilter $U$ extends a filter $F$, if $F \subseteq U$.}. 
\end{definition}

It is worth noting that $\gcon^\mathcal{C}$ reduces to $\models$ if $n=0$ or if $F(E(\mathcal{V}(\mathbb{T})))$ does not contain any grading propositions. However, unlike $\models$, $\gcon^\mathcal{C}$ is non-monotonic in general. In what follows, let $\mathbb{T}^\mathcal{C} = \{p ~|~\mathbb{T} \gcon^\mathcal{C} p\}$. When we are considering a
set of canons which only differ in the value of $n$, we write
$\mathbb{T}^n$ instead of $\mathbb{T}^\mathcal{C}$.

The upcoming example showcases the operation of $\G$ on the classical non-monotonic reasoning example of  birds fly, but penguins are special birds that do not fly. 

\begin{example}
\normalfont{
Consider the following $\G$ theory $\mathbb{T}_{OT1}$.
\begin{enumerate}\itemsep0pt\itemsep0pt
\item $\forall x [Bird(x) \supset \mathbf{G}(Flies(x),5)]$
\item $\forall x [Penguin(x) \supset \mathbf{G}(\neg
Flies(x),10)]$ \item $\forall x [Penguin(x) \supset Bird(x)]$
\item $Penguin(Opus)$ \item $Bird(Tweety)$
\end{enumerate}
We show next the relevant graded
consequences of $\mathbb{T}_{OT1}$ with respect to a series of canons, with $0 \leq n \leq 1$.

\begin{figure}[h]
  \begin{center}
  \fbox{\parbox{3in}{
  \small{
  $$
  \begin{array}{l|ll}
    n=0 \hspace{2em} & 0.1. & \mathbb{T}_\mathrm{OT1} \\
    & 0.2. & Bird(Opus)\\
        & 0.3. &  \mathbf{G}(Flies(Tweety),5)  \\
        & 0.4. &  \mathbf{G}(Flies(Opus),5)  \\
        & 0.5. & \mathbf{G}(\neg Flies(Opus),10)\\
        & & \\
    n=1  & 1.1. &  Everything~at~n=0\\
    & 1.2. & Flies(Tweety)\\
    & 1.3. & \neg Flies(Opus)
  \end{array}
  $$
  }}}
  \end{center}
  \label{fig:ot}
\end{figure}
Upon telescoping to $n=1$, we believe that Tweety flies and Opus
does not fly. The embedded proposition that Opus flies does not
survive telescoping since we trust that Opus does not fly, being a
penguin, more than we trust that it flies, being a bird.
$\mathbb{T}^1_{OT1}$ is a fixed point. 

Now, consider the
theory $\mathbb{T}_\mathrm{OT2}$ which is similar to
$\mathbb{T}_\mathrm{OT1}$, but with propositions (1) and (2) replaced
by ``$\mathbf{G}(\forall x [Bird(x) \supset Flies(x), 5)$''
and ``$\mathbf{G}(\forall x [Penguin(x) \supset \neg Flies(x),
10)$'', respectively. Thus, we trade the ``\emph{de re}''
representation of $\mathbb{T}_\mathrm{OT1}$ for the ``\emph{de
dicto}'' representation in $\mathbb{T}_\mathrm{OT2}$. This change
results in a change in the fixed point that we reach. In
 $\mathbb{T}^1_\mathrm{OT2}$, as in $\mathbb{T}^1_\mathrm{OT1}$, we end up
 believing that Opus does not fly. Unlike
 $\mathbb{T}^1_\mathrm{OT1}$ however, we give up our belief in the
 proposition that birds fly and, hence, cannot conclude that
 Tweety flies. Being able to grade only the consequent in $\G$, as in $\mathbb{T}_\mathrm{OT1}$, allows us to give up believing that Opus flies while keeping the rule $\forall x [Bird(x) \supset \mathbf{G}(Flies(x),5)]$ which allows to conclude that Tweety flies. Grading only part of the rule is one of the strengths of $\G$ which is not the possible in many weighted logics. \qed
 }
\end{example}

\section{Argument Systems} \label{argument-systems}
Argument systems \cite{lin1989argument} assume a logical language $\mathcal{L}$ which is a set of well-formed formulas (wffs) restricted to contain $\neg \phi$ if $\phi$ is itself a wff. The operator $\neg$ has no logical properties in argument systems. The most distinctive feature of argument systems is that they are based entirely on inference rules which are defined as follows.
\begin{definition}
An \emph{inference rule} is a rule of the following formats:
\begin{enumerate}\itemsep0pt
    \item $A$, where $A$ is a wff. This is called a \emph{base fact}.
    \item $A_1,...,A_m \rightarrow B$. This is called a \emph{monotonic rule}.
     \item $A_1,...,A_m \Rightarrow B$. This is called a \emph{non-monotonic rule}.
\end{enumerate}
\end{definition}
By chaining rules together, we get \emph{arguments} which are used to establish propositions.
\begin{definition}
Let $R$ be a set of rules. An argument in $R$ is a rooted tree with labelled arcs defined as follows:
\begin{enumerate}\itemsep0pt
    \item If $A$ is a base fact, the tree consisting of only $A$ as a root is an argument.
    \item If $p_1,...,p_m$ are arguments whose roots are $A_1,...,A_m$, and $A_1,...,A_m \rightarrow B~(A_1,...,A_m \Rightarrow B) \in R$ such that $B$ is not a node in the trees $p_1,...,p_m$, then the tree $p$ with $B$ as its root and $p_1,...,p_m$ as its immediate subtrees is an argument. All the arcs from B to its children is labelled by the monotonic (non-monotonic) rule.
    \end{enumerate}
\end{definition}
An argument $p$ is said to be for $\phi$ (or $\phi$ is supported by $p$) if $\phi$ is the root of $p$. By grouping arguments together, we get an \emph{argument structure}. An argument structure can be thought of as the set of logically consistent arguments held by the agent.
\begin{definition} \label{argstructure}
Let $R$ be a set of rules, an argument structure $T$ is defined as follows:
\begin{enumerate}\itemsep0pt
    \item if $p$ is a base fact, then $p \in T$.
    \item $\forall p \in T$, if $p'$ is a subtree of $p$, then $p' \in T$ (T is closed)
    \item if $p$ is formed from $p_1,...,p_n \in T$ by a monotonic rule, then $p \in T$ (T is monotonically closed).
    \item $\forall \phi \in \mathcal{L}$, $T$ does not contain arguments for both $\phi$ and $\neg \phi$ (T is consistent).
\end{enumerate}
\end{definition}
For argument structures, a notion of \emph{completeness} is defined with respect to a formula $\phi$ as follows.
\begin{definition}
An argument structure $T$ is \emph{complete} with respect to $\phi$ if $T$ contains an argument for either $\phi$ or $\neg \phi$.
\end{definition}
Finally, the belief space of an agent is defined as the set of formulas supported by an argument structure.
\begin{definition}
The set of formulas supported by an argument structure $T$, $Wff(T) = \{ \phi~|~ \exists p \in T$ such that $\phi$ is supported by $p\}$.
\end{definition}
The resulting framework made up of the inference rule, arguments, and argument structures is referred to as an \emph{argument system}.
\begin{example} \label{example:argument-rules}
\normalfont{ This example is due to \cite{lin1989argument}. Let $R$ be the following set of rules:

\begin{tabular}{c l}
$r1:$ &$true$ \\  
$r2:$ &$penguin(A)$ \\ 
$r3:$ &$penguin(A) \rightarrow bird(A)$\\
$r4:$ &$bird(A), \neg abnormal(bird(A)) \rightarrow flies(A)$\\  
$r5:$ &$penguin(A),\neg abnormal(penguin(A)) \rightarrow \neg flies(A)$\\  
$r6:$ &$penguin(A) \rightarrow abnormal(bird(A))$\\
$r7:$ &$true \Rightarrow \neg abnormal(penguin(A))$\\ 
$r8:$ &$true \Rightarrow \neg abnormal(bird(A))$ 
\end{tabular}

There are 8 possible arguments in $R$:

\begin{tabular}{c l}
$p1:$ &$true$ \\  
$p2:$ &$penguin(A)$ \\  
$p3:$ &$p2 \xrightarrow{r3} bird(A)$\\
$p4:$ &$p2 \xrightarrow{r6} abnormal(bird(A))$\\  
$p5:$ &$p1 \xRightarrow{r7} \neg abnormal(penguin(A))$\\ 
$p6:$ &$p1 \xRightarrow{r8} \neg abnormal(bird(A))$ \\
$p7:$ &$p3,p6 \xrightarrow{r4} flies(A)$\\ 
$p8:$ &$p2,p5 \xrightarrow{r5} \neg flies(A)$\\
\end{tabular}

Further, there are two possible argument structures: 

\begin{tabular}{c l}
$T_1:$ &$\{p1,p2,p3,p4\}$ \\  
$T_2:$ &$\{p1,p2,p3,p4,p5,p8\}$ \\  
\end{tabular}

Only $T_2$ is complete with respect to both $abnormal(bird(A))$ and $abnormal(penguin(A))$.\qed
}
\end{example}
\section{Argument Systems in $\G$}\label{argument-systems-logag}
In this section, we show how to encode argument systems in $\G$ theories, and prove that the $\G$ graded consequence relation can capture the notions of argument structures and supported propositions by an argument structure.
We start by presenting a mapping function from the inference rules of argument systems to $\G$ propositional terms. 
\begin{definition}
Let the mapping $\pi: R \rightarrow \sigma_\mathcal{P}$ from a set of inference rules $R$ of an argument system to a set of $\G$ propositional terms be defined as follows.
\begin{enumerate}\itemsep0pt
    \item If $A$ is a base fact, $\pi(A) = A$.
    \item If  $A_1,...,A_m \rightarrow B$ is a monotonic rule, $\pi(A_1,...,A_m \rightarrow B) = (\bigwedge_{i=1}^{m} A_i) \supset B$. 
    \item If  $A_1,...,A_m \Rightarrow B$ is a non-monotonic rule,  $\pi(A_1,...,A_m \Rightarrow B) = (\bigwedge_{i=1}^{m} A_i) \supset B$. 
\end{enumerate}
Whenever $S$ is a set of rules, $\pi(S)=\{\pi(\phi) ~|~ \phi \in S\}$.
\end{definition}
While the mapping function $\pi$ maps monotonic and non-monotonic rules to similar $\G$ propositional terms, the corresponding mappings will be treated differently when the corresponding $\G$ theory is constructed as will be shown below.

The following definition describes how to use the mapping function $\pi$ to construct $\G$ theories capable of capturing argument structures and supported propositions in argument structures. 

In the sequel, let the function $chain(\phi,d)$ mapping a rule $\phi$ to a $\G$ term denoting a grading proposition be defined as follows.
$$chain(\phi, d) = \left\{
\begin{array}{l r}
  \mathbf{G}(\phi,1) & \mathrm{if~} d=1 \\
  \mathbf{G}(chain(\phi,d-1),1) & \mathrm{otherwise}.\\
\end{array}
\right.$$
Further, let $\wp(\mathcal{S})$ be the set of non-empty subsets of a set $\mathcal{S}$. An indexing of $\wp(\mathcal{S})$ is a bijection $I: \wp(\mathcal{S}) \rightarrow \{1,2,...,|\wp(\mathcal{S})|\}$.

\begin{definition} \label{translation}
Let $R$ be a set of rules of an argument system with $R_M \subseteq R$ and $R_{N \hspace*{-1px} M} \subseteq R$ being the sets of monotonic and non-monotonic rules therein respectively, and let $I$ be an indexing of $\wp(R_{N \hspace*{-1px} M})$. The $I$-translation of $R$ to a $\G$ theory, referred to as $\mathbb{T}^I_{R}$, is the union of a monotonic subtheory $\mathbb{M}_{R}$, and a non-monotonic subtheory $\mathbb{NM}^I_{R}$. $\mathbb{M}_{R}$ is the smallest set satisfying the following: 
\begin{enumerate}\itemsep0pt
\item for all base facts $A \in R$, $\pi(A) \in \mathbb{M}_{R}$.
\item for all monotonic rules $r \in R$, $\pi(r) \in \mathbb{M}_{R} $.
\end{enumerate}
The non-monotonic subtheory $\mathbb{NM}^I_{R}$ is the smallest set satisfying the following:
\begin{enumerate}\itemsep0pt
    \item for all $\mathcal{S} \in \wp(R_{N \hspace*{-1px} M})$ and $r \in \mathcal{S}$, $chain(\pi(r), I(\mathcal{S})) \in \mathbb{NM}^I_{R}$.
    \item for all $\mathcal{S} \in \wp(R_{N \hspace*{-1px} M})$ and $r \not \in \mathcal{S}$,  $chain(\pi(r),I(\mathcal{S})) \in \mathbb{NM}^I_{R}$ and  $chain(\neg\pi(r),I(\mathcal{S})) \in \mathbb{NM}^I_{R}$.
\end{enumerate}
\end{definition}

\begin{example} \label{translation-a}
\normalfont{
Consider the set of rules $R$ of an argument system in Example \ref{example:argument-rules}. The monotonic subtheory $\mathbb{M}_R$ of the corresponding $\G$ translation is made of the following propositional terms: 

\begin{tabular}{c l}
$t1:$ &$true$ \\  
$t2:$ &$penguin(A)$ \\ 
$t3:$ &$penguin(A) \supset bird(A)$\\
$t4:$ &$bird(A) \wedge \neg abnormal(bird(A)) \supset flies(A)$\\  
$t5:$ &$penguin(A) \wedge \neg abnormal(penguin(A)) \supset \neg flies(A)$\\ 
$t6:$ &$penguin(A) \supset abnormal(bird(A))$\\
\end{tabular}

The following table shows the sets in $\wp(R_{N \hspace*{-1px} M})$ together with the output of \emph{a possible} indexing $I$.
\vspace*{2px}
\begin{center}
\begin{tabular}{|l | c | c |}
\hline
& &$I(\mathcal{S}_i)$ \\ \hline 
$\mathcal{S}_1$ &$\{true \Rightarrow \neg abnormal(penguin(A))\}$ &$1$  \\ \hline
$\mathcal{S}_2$ &$\{true \Rightarrow \neg abnormal(bird(A))\}$ &$2$ \\ \hline
$\mathcal{S}_3$ &$\{true \Rightarrow \neg abnormal(penguin(A)), true \Rightarrow \neg abnormal(bird(A))\}$ & $3$ \\ \hline
\end{tabular}
\end{center}
\vspace*{2px}

The $I$-translation to a $\G$ theory $\mathbb{T}_R$ is the union of $\mathbb{M}_R$ and $\mathbb{NM}^{I}_R$. The non-monotonic subtheory $\mathbb{NM}^{I}_R$ is made of the following terms:

\begin{tabular}{c l}
$t7:$ &$\mathbf{G}(true \supset \neg abnormal(penguin(A)),1)$ \\
$t8:$ &$\mathbf{G}(true \supset \neg abnormal(bird(A)),1)$ \\ 
$t9:$ &$\mathbf{G}(\neg (true \supset \neg abnormal(bird(A))),1)$ \\ 
$t10:$ &$\mathbf{G}(\mathbf{G}(true \supset \neg abnormal(bird(A)),1),1)$ \\ 
$t11:$ &$\mathbf{G}(\mathbf{G}(true \supset \neg abnormal(penguin(A)),1),1)$ \\ 
$t12:$ &$\mathbf{G}(\mathbf{G}(\neg (true \supset \neg abnormal(penguin(A))),1),1)$ \\ 
$t13:$ &$\mathbf{G}(\mathbf{G}(\mathbf{G}(true \supset \neg abnormal(penguin(A)),1),1),1)$\\
$t14:$ &$\mathbf{G}(\mathbf{G}(\mathbf{G}(true \supset \neg abnormal(bird(A)),1),1),1)$\\ 
\end{tabular} 
}
\end{example}\qed

The basic intuition behind this construction is to capture all the possible argument structures with one $\G$ theory by utilizing the notion of successive levels of graded consequences. The idea is to construct $\mathbb{T}^I_R$ in a way such that all the rules in any possible argument structure are graded consequences at some level $n$. This is accomplished by translating the base facts and monotonic rules to equivalent $\G$ propositional non-graded (hence, certain) terms. In this way, base facts and monotonic rules will be graded consequences at all levels. The non-monotonic rules will be represented as graded (hence, uncertain) $\G$ propositional terms. Since each argument structure contains any possible subset of non-monotonic rules, we embed each non-monotonic rule in each possible subset $\mathcal{S} \in \wp(R_{N \hspace*{-1px} M})$ in a $\G$ grading propositional term at an embedding degree of $I(\mathcal{S})$ (condition 1 of the construction of $\mathbb{NM}^{I}_R$). To make sure that the rules that are graded consequences at level $n$ are only the rules in $\mathcal{S}$, we embed any rule not in $\mathcal{S}$ as well as its negation in a $\G$ grading propositional term at an embedding degree of $I(\mathcal{S})$ (condition 2 of the construction of $\mathbb{NM}^{I}_R$). It is worth noting that any possible indexing $I$ is possible as long as each set in $\wp(R_{N \hspace*{-1px} M})$ is assigned a unique index so that all the rules in one set will be embedded at a different level from the rules in another set. 

Henceforth, let $R$ be a  set of inference rules for some argument system with $R_M \subseteq R$ and $R_{N \hspace*{-1px} M} \subseteq R$ being the sets of monotonic and non-monotonic rules therein respectively. Throughout, we assume that $R_M$ is \emph{consistent}. Further, let $I$ be an indexing of $\wp(R_{N \hspace*{-1px} M})$, $\mathbb{T}^{I}_{R}=\mathbb{M}_R \cup \mathbb{NM}^{I}_R$ be the $I$-translation to a $\G$ theory, and the set of $\mathbb{T}^{I}_{R}$ interpretations be $\mathcal{T}=\mathcal{V}(\mathbb{T}^{I}_{R})$. Finally, let $T$ be an argument structure where $R(T)$ is the set of all base facts in $T$ union the set of all the rules appearing as arc labels in all arguments in $T$. The following simple observation follows directly from our construction.
\begin{obs} \label{structure-rules}
Let $R$ be a set of rules of an argument system. For any argument structure $T$,\\ $R(T) = R^T_M \cup \mathcal{S}$ where $R^T_M \subseteq R_M$ and either $\mathcal{S}=\varnothing$ if $R(T)$ contains no non-monotonic rules,  or $\mathcal{S}\in \wp(R_{N \hspace*{-1px} M})$ if $R(T)$ contains at least one non-monotonic rule where $\mathbb{M}_R \cup \mathcal{S}$ is consistent. 
\end{obs}

The following observation states that, in our construction, the fused grade of $\pi(r)$ at level $n$ has a fused of grade $\mathfrak{f}_\mathfrak{T}(\interp{\pi(r)}^\mathcal{V},E^n(\mathcal{T}))$ of $n$ if we choose the grade fusion operator $\otimes=sum$ and $\oplus=max$. 

\begin{obs}\label{fused-grade}
For any telescoping structure $\mathfrak{T}=\langle \mathcal{T}, \mathfrak{O}, \otimes, \oplus\rangle$ with $\otimes = sum$ and $\oplus = max$,\\ if $chain(\pi(r),n) \in \mathbb{NM}^{I}_R$, $\mathfrak{f}_\mathfrak{T}(\interp{\pi(r)}^\mathcal{V},E^n(\mathcal{T})) = n$. 
\end{obs}
\begin{proof}
Taking the $\otimes$ as the sum operator and the $\oplus$ as the max operator, since all grades in all grading chains in $\mathbb{NM}^{I}_R$ are 1s, then the interpretation of the embedded rule $\interp{\pi(r)}^\mathcal{V}$ at level $n$ has a grade of $n$. This is because, even if $\pi(r)$ appears graded in shorter chains, the $\oplus$ operator will set the grade of $\interp{\pi(r)}^\mathcal{V}$ to the deeper depth $n$. 
\end{proof}
It follows directly from Observation \ref{fused-grade} that the rules with higher embedding degrees have higher grades than the rules with lower embedding degrees.

We introduce the following notation for the ease of readability of our upcoming proofs. Let the set of \emph{embedded graded rules} be $ER_n=\{\interp{\pi(r)}^\mathcal{V} ~|~  chain(\pi(r),n) \in \mathbb{NM}^{I}_R~and~ chain(\pi(\neg r),n) \not\in \mathbb{NM}^{I}_R\}$, and $G_{\mathcal{T}}$ be the set of all \emph{grading propositions} in $E^n(\mathcal{T})$. The following observation states that if the filter of the base facts, the monotonic rules (which we originally assume to be consistent), and the non-monotonic rules embedded at level $n$ is consistent, then an embedded rule $\interp{\pi(r)}^\mathcal{V}$ at level $n$ while $\interp{\pi(\neg r)}^\mathcal{V}$ is not embedded at the same level is a member of the graded filter of degree $n$. The intuition is that $\interp{\pi(r)}^\mathcal{V}$ must have a higher grade given our construction and Observation \ref{fused-grade}. 

\begin{obs}\label{survivors}
For any telescoping structure $\mathfrak{T}=\langle \mathcal{T}, \mathfrak{O}, \otimes, \oplus\rangle$ with $\otimes = sum$ and $\oplus = max$,
\\ if $F(\mathcal{V}(\mathbb{M_R} \cup ER_n)$ is consistent ($\not=\mathcal{P})$ and $\interp{\pi(r)}^\mathcal{V} \in ER_n$, then $\interp{\pi(r)}^\mathcal{V} \in \mathcal{F}^n(\mathfrak{T})$ and $\interp{\neg \pi(r)}^\mathcal{V} \not\in \mathcal{F}^n(\mathfrak{T})$.
\end{obs}
\begin{proof}
Suppose that $F(\mathcal{V}(\mathbb{M_R} \cup ER_n)$ is consistent.
If $chain(\pi(r),n) \in \mathbb{NM}^{I}_R$ and $chain(\neg \pi(r),n) \not \in \mathbb{NM}^{I}_R$, then
by Observation \ref{fused-grade}, the fused grade of $\interp{\pi(r)}^\mathcal{V}$ in $E^n(\mathcal{T})$ is $n$. If $\interp{\pi(\neg r)}^\mathcal{V} \not\in E^n(F(\mathcal{T}))$, then $\interp{\pi(r)}^\mathcal{V}$ survives telescoping and is supported at level $n$ and $\interp{\pi(r)}^\mathcal{V} \in \mathcal{F}^n(\mathfrak{T})$. Otherwise, if $\interp{\pi(\neg r)}^\mathcal{V} \in E^n(F(\mathcal{T}))$, we have three cases.
\begin{enumerate}
    \item $\interp{\pi(\neg r)} ^\mathcal{V} \in F(\mathcal{T})$. But this implies that $F(\mathcal{V}(\mathbb{M}_R)~ \cup~ER_n)$ is inconsistent. Hence, we get a contradiction; or
    \item $\interp{\pi(\neg r)}^\mathcal{V}$ is embedded in a grading chain of length $n$. This can not be as $chain(\neg \pi(r),n) \not \in \mathbb{NM}^{I}_R$; or
    \item $\interp{\pi(\neg r)}^\mathcal{V}$ is supported by some graded propositions embedded at a degree of at most $n$. If $\interp{\pi(\neg r)}^\mathcal{V}$ is supported by at least a graded proposition with an embedding degree of $n$, we get a contradiction as $F(\mathcal{V}(\mathbb{M}_R)~\cup ~ER_n)$ must be inconsistent. Then, it must be that $\interp{\pi(\neg r)}^\mathcal{V}$ is supported by graded propositions of embedding degrees less than $n$. In this case, however, all such graded propositions will not survive telescoping as $\interp{\pi(\neg r)}^\mathcal{V}$ has a higher grade depriving $\interp{\pi(\neg r)}^\mathcal{V}$ of its support. It follows then that $\interp{\pi(r)}^\mathcal{V}$ survives telescoping and is supported at level $n$. Hence, $\interp{\pi(r)}^\mathcal{V} \in \mathcal{F}^n(\mathfrak{T})$. 
\end{enumerate}
\end{proof}

The following proposition will prove to be very useful in the remaining of this section. It states that if the filter of the base facts, the monotonic rules (which we originally assume to be consistent), and the non-monotonic rules embedded at level $n$ is consistent, then the graded filter of degree $n$ is equal to the valuation of the monotonic subtheory, all the embedded rules in $ER_n$, and all the grading propositions in $G_\mathcal{T}$.

\begin{proposition} \label{prop:n-filter}
If $F(\mathcal{V}(\mathbb{M}_R)\cup~ER_n)$ is consistent, then $\mathcal{F}^n(\mathfrak{T})=F(\mathcal{V}(\mathbb{M}_R)\cup~ ER_n~\cup~G_\mathcal{T}) $ for every telescoping structure $\mathfrak{T}=\langle \mathcal{T}, \mathfrak{O}, sum, max \rangle$.
\end{proposition}
\begin{proof}
We prove this by showing that $F(\mathcal{V}(\mathbb{M}_R)~\cup~ ER_n \cup~ G_\mathcal{T}) \subseteq \mathcal{F}^n(\mathfrak{T})$ and $\mathcal{F}^n(\mathfrak{T})\subseteq F(\mathcal{V}(\mathbb{M}_R)~\cup~ ER_n~\cup~\\G_\mathcal{T})$. Suppose that $F(\mathcal{V}(\mathbb{M}_R)~ \cup~ER_n)$ is consistent. Let $\interp{\pi(\phi)}^\mathcal{V} \in \mathcal{V}(\mathbb{M}_R)~\cup~ ER_n ~\cup~ G_\mathcal{T}$. Therefore, it must be one of the following cases.
\begin{enumerate}\itemsep0pt
    \item $\interp{\pi(\phi)}^\mathcal{V} \in \mathcal{V}(\mathbb{M}_R)$. In this case, $\interp{\pi(\phi)}^\mathcal{V}$ survives telescoping and is supported at level $n$ since all members of the top theory $\mathcal{T}$ survive telescoping and are supported at all levels.
    \item $\interp{\pi(\phi)}^\mathcal{V} \in ER_n$. By Observation \ref{survivors}, $\interp{\pi(\phi)}^\mathcal{V}$ survives telescoping and is supported at level $n$.
\item $\interp{\pi(\phi)}^\mathcal{V} \in G_{\mathcal{T}}$. The only possible grading propositions come from $\mathbb{NM}^I_{R}$ or embedded grading propositions in $\mathbb{NM}^I_{R}$. The interpretations of such grading propositions must survive telescoping and are supported at level $n$ since such grading propositions are never members of $\bot$-kernels given the construction of $\mathbb{NM}^I_{R}$. 
\end{enumerate}
Hence, $\mathcal{V}(\mathbb{M}_R) \cup ER_n \cup G_\mathcal{T}$ survive telescoping and is supported at level $n$. By the definition of graded filters then, $F(\mathcal{V}(\mathbb{M}_R)~\cup~ER_n~\cup~G_\mathcal{T}) \subseteq \mathcal{F}^n(\mathfrak{T})$.

Now, we proceed to proving that $\mathcal{F}^n(\mathfrak{T})\subseteq F(\mathcal{V}(\mathbb{M}_R)~\cup~ ER_n~\cup~G_\mathcal{T})$. According to Observation \ref{supported-filter} and the definition of graded filters, $\mathcal{F}^n(\mathfrak{T})=F(F(\mathcal{T})~\cup~EG) =  F(\mathcal{T}~\cup~EG)= F( \mathcal{V}(\mathbb{M}_R)~\cup~\mathcal{V}(\mathbb{NM}^{I}_{R})~\cup~ EG) $ for some set of embedded graded propositions $EG$ in $E^n(\mathcal{T})$ that survive telescoping and are supported at level $n$. Let $\interp{\pi(\phi)}^\mathcal{V} \in \mathcal{V}(\mathbb{M}_R)~\cup~\mathcal{V}(\mathbb{NM}^{I}_{R})~ \cup~ EG $. Therefore, one of the following cases is true.
\begin{enumerate}\itemsep0pt
    \item $\interp{\pi(\phi)}^\mathcal{V}\in \mathcal{V}(\mathbb{M}_R)$. It follows trivially then that $\interp{\pi(\phi)}^\mathcal{V} \in \mathcal{V}(\mathbb{M}_R)~\cup~ ER_n~\cup~G_\mathcal{T}$.
    \item $\interp{\pi(\phi)}^\mathcal{V} \in \mathcal{V}(\mathbb{NM}^{I}_{R})$. Hence, $\interp{\pi(\phi)}^\mathcal{V} \in G_\mathcal{T}$ by the definition of $G_\mathcal{T}$ and the construction of $\mathbb{NM}^{I}_{R}$.
    \item $\interp{\pi(\phi)}^\mathcal{V} \in EG$. If $\interp{\pi(\phi)}^\mathcal{V}$ is a grading proposition, then $\interp{\pi(\phi)}^\mathcal{V} \in G_\mathcal{T}$. Otherwise, $\interp{\pi(\phi)}^\mathcal{V}$ is a graded proposition, then according to Definition \ref{translation} either $chain(\pi(\phi),n) \in \mathbb{NM}^{I}_{R}$ and\\ $chain(\neg \pi(\phi),n) \not\in \mathbb{NM}^{I}_{R}$ or $chain(\pi(\phi),n) \in \mathbb{NM}^{I}_{R}$ and $chain(\neg \pi(\phi),n) \in \mathbb{NM}^{I}_{R}$. However, in the second case by Observation \ref{fused-grade} both $\interp{\pi(\phi)}^\mathcal{V}$ and $\interp{\pi(\neg \phi)}^\mathcal{V}$ have the same grade of $n$ in $E^n(\mathcal{T})$. Accordingly, both do not survive telescoping at level $n$. It must then be that $chain(\pi(\phi),n) \in \mathbb{NM}^{I}_{R}$ and $chain(\neg \pi(\phi),n) \not\in \mathbb{NM}^{I}_{R}$ for $\interp{\pi(\phi)}^\mathcal{V}$ to survive telescoping at level $n$ according to Observation \ref{survivors}. Hence, $\interp{\pi(\phi)}^\mathcal{V} \in ER_n$.
\end{enumerate}
Thus, $\mathcal{V}(\mathbb{M}_R)~\cup~\mathcal{V}(\mathbb{NM}^{I}_{R})~\cup~ EG \subseteq \mathcal{V}(\mathbb{M}_R)~\cup~ ER_n~ \cup~G_\mathcal{T}$. Since filters are monotonic,\\ then $F(\mathcal{V}(\mathbb{M}_R)~\cup~\mathcal{V}(\mathbb{NM}^{I}_{R})~\cup~ G) \subseteq F(\mathcal{V}(\mathbb{M}_R)~\cup~ ER_n~\cup~G_\mathcal{T})$. Hence, $\mathcal{F}^n(\mathfrak{T}) = F(\mathcal{V}(\mathbb{M}_R)~\cup~ ER_n~ \cup~G_\mathcal{T})$.
\end{proof}
Having stated the previous proposition, we now go back to Example \ref{translation-a} to show the sequence of graded consequences at successive levels.
\begin{example}
\normalfont{
We show the relevant graded consequences of $\mathbb{T}_R^I$ in Example \ref{translation-a} with respect to a series of canons, with $0 \leq n \leq 3$ with $\otimes=sum$ and $\oplus=max$.
\begin{figure}[h]
  \begin{center}
  \fbox{\parbox{3.5in}{
  \small{
  $$
  \begin{array}{l|ll}
    n=0 \hspace{0.2em} & 0.1. & \mathbb{T}_R^I \\
    & 0.2.  & bird(A) \\
    & 0.3.  & abnormal(bird(A))\\
    & & \\
    n=1  & 1.1. &Everything~at~n=0\\
    & 1.2. & true \supset \neg abnormal(penguin(A))\\
    & 1.3. & \mathbf{G}(true \supset \neg abnormal(bird(A)),1)\\
    & 1.4. & \mathbf{G}(true \supset \neg abnormal(penguin(A)),1)\\
    & 1.5.  & \mathbf{G}(\neg (true \supset \neg abnormal(penguin(A))),1)\\
    & 1.6.  & \mathbf{G}(\mathbf{G}(true \supset \neg abnormal(penguin(A)),1),1)\\
    & 1.7.  & \mathbf{G}(\mathbf{G}(true \supset \neg abnormal(bird(A)),1),1)\\
    & 1.8. &  \neg abnormal(penguin(A)) \\
    & 1.9. &  \neg flies(A) \\
    & & \\
    n=2  & 2.1. &  Everything~at~n=1~except~1.2,~1.8~and~1.9\\
    & 2.2. & \mathbf{G}(true \supset \neg abnormal(penguin(A)),1)\\
    & 2.3. & \mathbf{G}(true \supset \neg abnormal(bird(A)),1)\\
    & & \\
    n=3  & 3.1. &  Everything~at~n=2\\
     & 3.2. & true \supset \neg abnormal(penguin(A))\\
     & 3.3. &  \neg abnormal(penguin(A)) \\
    & 3.4. &  \neg flies(A) \\
  \end{array}
  $$
  }}}
  \end{center}
  \label{fig:ot}
\end{figure}

At $n=0$, we get 0.2 as it follows from $t2$ and $t3$ and 0.3 as it follows from $t2$ and $t6$ (the terms are shown in Example \ref{translation-a}). This level of graded consequences corresponds to the supported wffs in the first argument structure in Example \ref{example:argument-rules} (the one containing no non-monotonic rules). 

Upon telescoping to $n=1$, we get all the graded propositional terms embedded at level 1 in $\mathbb{NM}^{I}_R$ (1.2 to 1.7) by Observation \ref{survivors} since $F(\mathcal{V(\mathbb{M}_R)} \cup ER_1)$ is consistent. 
As a consequence, we get 1.8 as it follows from 1.2 and $t1$ and 1.9 as it follows from $t2$, $t5$, and 1.8. Therefore, at level 1 we end up believing that $A$ is not an abnormal penguin that does not fly. This level of graded consequences corresponds to the supported wffs in the second argument structure in Example \ref{example:argument-rules} (the one containing $r7$) since $\pi(r7)$ is a graded consequence of $\mathbb{T}_R^I$ at level $1$.

Going to level 2, both $true \supset \neg abnormal(penguin(A))$ and $\neg(true \supset \neg abnormal(penguin(A)))$ are extracted with a grade of $2$ (recall that embedded rules at level $n$ have a grade of $n$ according to Observation \ref{fused-grade}). Accordingly, both do not survive telescoping and 1.2 goes away depriving $1.8$ and $1.9$ of their support. The embedded graded propositional term $true \supset \neg abnormal(bird(A))$ in 1.3 does not survive telescoping as well as it contradicts with $0.3$. Since the interpretation of 0.3 is in the filter of the top theory, $true \supset \neg abnormal(bird(A))$ is kicked out. At level 2, we end up not believing that $A$ is an abnormal penguin nor do we believe that it flies. This level of graded consequences corresponds, just like the graded consequences at level 0, to the supported wffs in the first argument structure in Example \ref{example:argument-rules}.

Upon telescoping to level $3$, $true \supset \neg abnormal(penguin(A))$ comes back as it now has a higher grade (3) than $\neg (true \supset \neg abnormal(penguin(A)))$ (2). Accordingly, we get $\neg abnormal(penguin(A))$ and $\neg flies(A)$ back at level $3$. The embedded proposition at level $3$ $true \supset \neg abnormal(bird(A))$ still does not survive telescoping as it contradicts with 0.3. Thus, at level 3, we go back to believing that $A$ is not an abnormal penguin that does not fly. This level of graded consequences corresponds, just like the graded consequences at level 1, to the supported wffs in the first argument structure in Example \ref{example:argument-rules}.} \qed
\end{example}

As illustrated by this example, the set of graded consequences at any level $n$ correspond to the supported wffs in some argument structure. Whenever $F(\mathcal{V(\mathbb{M}_R)} \cup ER_n)$ is inconsistent, the inconsistency will be resolved by kicking out the rules with the least grade. As a result, the set of graded consequences at such level $n$ corresponds to the supported wffs in some argument structure that uses only the consistent rules. This is demonstrated at levels 2 and 3 of the previous example.

In the remaining of this section, we prove that using the proposed translation we can capture the notion of supported wffs in argument structures. We start by proving that the interpretation of the translation of any base fact or any wff supported by entirely monotonic rules in $R$ is a member of the graded filters of all degrees $n$.

\begin{lemma} \label{lemma1}
For all base facts $A \in T$ and all degrees $n$, it must be that $\interp{\pi(A)}^\mathcal{V} \in \mathcal{F}^n({\mathfrak{T}})$ and for every telescoping structure $\mathfrak{T}=\langle \mathcal{T}, \mathfrak{O}, \otimes, \oplus \rangle$.
\end{lemma}
\begin{proof}
Any base fact $A \in R$ (and hence in $T$ by the definition of argument structures), $\pi(A)$ is ungraded in the monotonic subtheory $\mathbb{M}_R$ according to Definition \ref{translation}.  It follows from Observation \ref{supported-filter} and the definition of graded filters that $\interp{\pi(A)}^{\mathcal{V}} \in \mathcal{F}^n({\mathfrak{T}})$ for every telescoping structure $\mathfrak{T}=\langle \mathcal{T}, \mathfrak{O}, \otimes, \oplus \rangle$ since all members of the top theory $\mathcal{T}$ survive telescoping and are supported at all levels.
\end{proof}

\begin{lemma} \label{lemma2}
Let $\phi$ be the root of some argument $\phi \in T$ where all the arcs in $\phi$ are labelled by monotonic rules. It must be that $\interp{\pi(\phi)}^{\mathcal{V}} \in \mathcal{F}^n({\mathfrak{T}})$ for all degrees $n$ and for every telescoping structure $\mathfrak{T}=\langle \mathcal{T}, \mathfrak{O}, \otimes, \oplus \rangle$.  
\end{lemma}
\begin{proof}
We prove this by induction on the height of the argument tree. \\
\textbf{Base case:} The only argument trees containing single nodes in $T$ are the base facts. Hence, the base case follows from Lemma $\ref{lemma1}$.\\
\textbf{Induction hypothesis:} If $\phi$ is the root of some argument tree $\phi \in T$ of height at most $h$, then $\interp{\pi(\phi)}^{\mathcal{V}} \in \mathcal{F}^n({\mathfrak{T}})$ for all $n$ and for every telescoping structure $\mathfrak{T}=\langle \mathcal{T}, \mathfrak{O}, \otimes, \oplus \rangle$.\\
\textbf{Induction step:} Suppose that $\phi$ is the root of some argument tree $\phi \in T$ of height $h+1$. Then, there must be a monotonic rule $ r = A_1,...,A_m \rightarrow \phi \in R$ and, since $T$ is monotonically closed, $A_1,...,A_m$ are roots of argument trees of height at most $h$ in $T$. Hence, by the induction hypothesis, $\interp{\pi(A_i)}^{\mathcal{V}} \in \mathcal{F}^n({\mathfrak{T}})$  for $1 \le i \le m$, all $n$, and every telescoping structure $\mathfrak{T}=\langle \mathcal{T}, \mathfrak{O}, \otimes, \oplus \rangle$.  According to Definition \ref{translation}, $\pi(r) \in \mathbb{M}_{R}$ and, hence, $\interp{\pi(r)}^\mathcal{V} \in \mathcal{F}^n({\mathfrak{T}})$ for all $n$ by the definition of graded filters since all members of the top theory $\mathcal{T}$ survive telescoping and are supported at all levels. It follows then that $\interp{\pi(\phi)}^{\mathcal{V}} \in \mathcal{F}^n({\mathfrak{T}})$ as well. 
\end{proof}

It remains to prove that any wff supported by at least one non-monotonic rule is a member of a graded filter of some degree $n$. In order to prove this, we need the following proposition. 

\begin{proposition}\label{prop:rules}
Let $R(T) = R^T_M \cup \mathcal{S}$ according to Observation \ref{structure-rules}.
For all $\phi \in R(T)$, $\pi(\phi) \in \mathcal{F}^n(\mathfrak{T})$ for every telescoping structure $\mathfrak{T}=\langle \mathcal{T}, \mathfrak{O}, sum,max \rangle$ for all $n$ if $\mathcal{S}=\varnothing$ and for $n=I(\mathcal{S})$ otherwise.
\end{proposition}
\begin{proof}
Let $\phi \in R(T)$. If $\mathcal{S}=\varnothing$, then $R(T)=R^T_M$ is made up of only monotonic rules and $\pi(\phi) \in \mathbb{M}_R$ according to Definition \ref{translation}. By the definition of graded filters, $\interp{\pi(\phi)}^\mathcal{V}\in \mathcal{F}^n(\mathfrak{T})$ for all $n$. Otherwise, if $\mathcal{S}\not= \varnothing$, then according to Definition \ref{translation}, $chain(r,I(\mathcal{S})) \in  \mathbb{NM}^{I}_{R}$ and  $chain(\neg r,I(\mathcal{S})) \not \in  \mathbb{NM}^{I}_{R}$. Since $T$ is an argument structure, it must be that $R(T)=R^T_M\cup\mathcal{S}$ is consistent. Accordingly, $\pi(R(T))$ must be consistent as well. Hence, it follows from Proposition \ref{prop:n-filter} that $\interp{\pi(\phi)}^\mathcal{V} \in \mathcal{F}^n(\mathfrak{T})$ with $n=I(\mathcal{S})$. 
\end{proof}

\begin{lemma} \label{lemma3}
Let $\phi$ be the root of some argument $\phi \in T$ where, if $\phi$ has arcs, there is at least one arc in $\phi$ labelled by a non-monotonic rule.
Further, let $n = I(\mathcal{S})$ where $\mathcal{S} \in \wp(R_{N \hspace*{-1px} M})$ contains all the non-monotonic rules in $\phi$. It must be that $\interp{\pi(\phi)}^{\mathcal{V}} \in \mathcal{F}^n({\mathfrak{T}})$ for every telescoping structure $\mathfrak{T}=\langle \mathcal{T}, \mathfrak{O}, sum, max \rangle$.
\end{lemma}
\begin{proof}
We prove this by induction on the height of the argument tree.\\
\textbf{Base case:} The only argument trees containing single nodes in $T$ are the base facts. Hence, the base case follows from Lemma $\ref{lemma1}$.\\
\textbf{Induction hypothesis:} If $\phi$ is the root of some argument tree $\phi \in T$ of height at most $h$, then $\interp{\pi(\phi)}^{\mathcal{V}} \in \mathcal{F}^n({\mathfrak{T}})$ for some degree $n = I(\mathcal{S})$ and for every telescoping structure $\mathfrak{T}=\langle \mathcal{T}, \mathfrak{O}, sum, max \rangle$.\\
\textbf{Induction step:} Suppose that $\phi$ is the root of some argument tree $\phi \in T$ of height $h+1$ with direct children $A_1,...,A_m$. Since $T$ is closed, $A_1,...,A_m$ are roots of argument trees in $T$ of height at most $h$. Hence, by the induction hypothesis, $\interp{\pi(A_i)}^{\mathcal{V}} \in \mathcal{F}^n({\mathfrak{T}})$ for $1\le i \le m$, $n = I(\mathcal{S})$, and every telescoping structure $\mathfrak{T}=\langle \mathcal{T}, \mathfrak{O}, sum, max \rangle$. By Proposition \ref{prop:rules}, $\pi(R(T)) \subset \mathcal{F}^n(\mathfrak{T})$ including the rules labelling the arcs from $A_1,...A_m$ to $\phi$. It follows then that $\interp{\pi(\phi)}^{\mathcal{V}} \in \mathcal{F}^n(\mathfrak{T})$ for every telescoping structure $\mathfrak{T}=\langle \mathcal{T}, \mathfrak{O},\\ sum, max \rangle$.
\end{proof}

We now use Lemmas \ref{lemma1}, \ref{lemma2}, and \ref{lemma3} to relate the notions of graded consequence in $\G$ and supported wffs in argument structures.

\begin{theorem} \label{theorem1}
For any argument structure $T$ and $n=I(\mathcal{S})$ where $\mathcal{S} \in \wp(R_{N \hspace*{-1px} M})$ contains all the non-monotonic rules appearing as arc labels in $T$, if $\phi \in Wff(T)$, then $\mathbb{T}^{I}_{R} \gcon^\mathcal{C} \pi(\phi)$ for some grading canon $\mathcal{C}=\langle \otimes, \oplus,n\rangle$ where $\otimes=sum$ and $\oplus=max$.
\end{theorem}
\begin{proof}
Suppose that $\phi \in Wff(T)$, then by the definition of $Wff(T)$ it must be one of the following three cases.
\begin{enumerate}\itemsep0pt
    \item $\phi$ is a base fact. By lemma 1, $\interp{\pi(\phi)}^{\mathcal{V}} \in \mathcal{F}^n({\mathfrak{T}})$ for all $n$ and every telescoping structure $\mathfrak{T}=\langle \mathcal{T}, \mathfrak{O}, \otimes, \oplus \rangle$. Hence, by Definition \ref{def:gfilter}, $\mathbb{T}^{I}_{R} \gcon^\mathcal{C} \pi(\phi)$ for some grading canon $\mathcal{C}=\langle sum, max,n\rangle$ where $n=I(\mathcal{S})$. 
    \item $\phi$ is a root of some argument with all arcs labelled by monotonic rules. By lemma 2, $\interp{\pi(\phi)}^{\mathcal{V}} \in \mathcal{F}^n({\mathfrak{T}})$ for all $n$ and every telescoping structure $\mathfrak{T}=\langle \mathcal{T}, \mathfrak{O}, \otimes, \oplus \rangle$. Hence, by Definition \ref{def:gfilter}, $\mathbb{T}^{I}_{R} \gcon^\mathcal{C} \pi(\phi)$ for some grading canon $\mathcal{C}=\langle sum, max,n\rangle$ where $n=I(\mathcal{S})$.
    \item $\phi$ is a root of some argument with all arcs labelled by non-monotonic rules. By lemma 3, $\interp{\pi(\phi)}^{\mathcal{V}} \in \mathcal{F}^n({\mathfrak{T}})$ for $n=I(\mathcal{S})$ and every telescoping structure $\mathfrak{T}=\langle \mathcal{T}, \mathfrak{O}, sum, max, \rangle$. Hence, by Definition \ref{def:gfilter}, $\mathbb{T}^{I}_{R} \gcon^\mathcal{C} \pi(\phi)$ for some grading canon $\mathcal{C}=\langle sum, max,n\rangle$.
\end{enumerate}
\end{proof}

Since the results in \cite{lin1989argument} showing that default logic, autoepistemic logic, circumscription, and the principle of negation as failure are all special cases of argument systems utilize the notion of completeness with respect to a wff $\phi$, we need to relate the same notion to our notion of graded consequence. The following corollary does exactly that and follows directly from Theorem \ref{theorem1}.

\begin{cor}
If $T$ is complete with respect to $\phi$, then $\mathbb{T}^{I}_{R} \gcon^\mathcal{C} \phi$ or $\mathbb{T}_{R} \gcon^\mathcal{C} \neg \phi$ for some grading canon $\mathcal{C}=\langle \otimes, \oplus,n\rangle$ where $\otimes=sum$ and $\oplus=max$.
\end{cor}

It turns out, however, that if $\mathbb{T}^{I}_{R} \gcon^\mathcal{C} \pi(\phi)$ with $\mathcal{C}=\langle sum, max,n\rangle$ and $n=I(\mathcal{S})$, it is not necessarily the case that $\phi \in Wff(T)$. This is mainly because, given our construction in Definition \ref{translation}, the rules in $R(T)$ are graded consequences of $\mathbb{T}^{I}_{R}$ at level $n$ but are not in $Wff(T)$ as the rules appearing as arc labels in $T$ are never in $Wff(T)$. Further, our construction maps all monotonic rules in $R$ including the monotonic rules not in $R(T)$ to $\G$ propositional terms in $\mathbb{M}_R$. Such monotonic rules will also be graded consequences of $\mathbb{T}^{I}_{R}$ at level $n$ but are not in $Wff(T)$. For this reason, Theorem \ref{theorem2} presents a more general result. We prove that if $\mathbb{T}^{I}_{R} \gcon^\mathcal{C} \pi(\phi)$ then $\pi(\phi)$ must be a logical consequence of $\pi(R(T) \cup R_M')$ where $R_M'$ is a maximal subset of $R_M$ consistent with $R(T)$. It is important to note here that $\pi(Wff(T))$ is a subset of the set of logical consequences of  $\pi(R(T) \cup R_M')$.

\begin{theorem}\label{theorem2}
For any argument structure $T$ and $n=I(\mathcal{S})$ where $\mathcal{S} \in \wp(R_{N \hspace*{-1px} M})$ contains all the non-monotonic rules appearing as arc labels in $T$, if $\mathbb{T}^{I}_{R} \gcon^\mathcal{C} \pi(\phi)$ for a grading canon $\mathcal{C}=\langle sum, max,n\rangle$, then $\pi(\phi)$ is a logical consequence of $\pi(R(T) \cup R_M')$.
\end{theorem}
\begin{proof}
Suppose that $\mathbb{T}^{I}_{R} \gcon^\mathcal{C} \pi(\phi)$ with $\mathcal{C}=\langle sum, max,n\rangle$. Then, $\interp{\pi(\phi)}^\mathcal{V} \in \mathcal{F}^n(\mathfrak{T})$. According to Proposition \ref{prop:n-filter}, $\mathcal{F}^n(\mathfrak{T}) = F(\mathcal{V}(\mathbb{M}_R)~\cup~ ER_n~\cup~G_\mathcal{T})$. We have four cases.
\begin{enumerate}\itemsep0pt
    \item $\interp{\pi(\phi)}^\mathcal{V} \in  \mathcal{V}(\mathbb{M}_R)$. In this case, $\phi$ is either a base fact or a monotonic rule in $R$. If $\phi$ is a base fact, then it must be in $Wff(T)$ by the definition of argument structures (and hence $\pi(\phi)$ a logical consequence of $R(T)$). If $\phi$ is a monotonic rule appearing as an arc label in $T$, then $\phi \in R(T)$. Otherwise, if $\phi$ is a monotonic rule that does not appear as an arc label in $T$, then $\phi \in R_M'$. In the three cases, $\pi(\phi)$ is a logical consequence of $\pi(R(T) \cup R_M')$.
    \item $\interp{\pi(\phi)}^\mathcal{V}  \in ER_n$. In this case, $\pi(\phi)$ must be a non-monotonic rule embedded at level $n$ whose negation is not embedded at level $n$. Hence, $\phi \in \mathcal{S}$ and $\phi \in R(T)$ by the definition of $R(T)$. Hence, $\pi(\phi) \in \pi(R(T))$.
    \item $\interp{\pi(\phi)}^\mathcal{V} \in G_\mathcal{T}$. This can not be as $\phi$ must be a grading term and grading terms are never in $R$.
    \item $\interp{\phi}^\mathcal{V} \in  F(\mathcal{V}(\mathbb{M}_R)~\cup~ ER_n~\cup~G_\mathcal{T})$. It follows from the previous three cases and the monotonicity of filters that $\pi(\phi)$ is a logical consequence of $\pi(R(T) \cup R_M')$.
\end{enumerate}
\end{proof}
\section{Conclusion} \label{conc}
The purpose of this paper was to present an algebraic graded non-monotonic logic we refer to as $\G$ capable of encompassing a wide variety of non-monotonic logical formalisms. We showed how argument systems can be encoded in $\G$, and proved that the $\G$ logical consequence relation captures the notion of supported propositions in argument structures. Since default logic, autoepistemic logic, circumscription, the principle of negation as failure are all proved to be special cases of argument systems, our results show that $\G$ captures such the previously-mentioned non-monotonic logics as well. Previous results show that $\G$ subsumes possibilistic logic and any non-monotonic inference relation satisfying Makinson's rationality postulates. This proves the universality of $\G$ as a non-monotonic logic. To the best of our knowledge, $\G$ is the only framework in the literature that was shown to capture weighted approaches to non-monotonicity such as possibilistic logic in addition to the classical previously- mentioned approaches. In this way, $\G$ provides a powerful unified framework for non-monotonicity. 

\bibliographystyle{eptcs}
\bibliography{generic}
\end{document}